\documentclass[sn-mathphys-num]{sn-jnl}

\usepackage[T1]{fontenc}
\usepackage[utf8]{inputenc}
\usepackage{amsmath,amssymb,amsfonts,mathtools}
\usepackage{booktabs}

\theoremstyle{thmstyleone}
\newtheorem{theorem}{Theorem}[section]
\newtheorem{lemma}[theorem]{Lemma}

\theoremstyle{thmstylethree}

\newcommand{\trans}{\operatorname{tr}}
\renewcommand{\proofname}{Proof}

\begin{document}
\renewcommand{\proofname}{Proof}
\renewcommand{\abstractname}{Abstract}
\emergencystretch=.8em

\title[On the Offline Version of the Time-Optimal k-Server Problem]{On the Offline Version of the
Time-Optimal \texorpdfstring{$k$}{k}-Server Problem}

\author*[1]{\fnm{Oleg} \sur{Lomachenko}}\email{lomachenko.science@gmail.com}

\affil*[1]{\orgname{GoodsForecast}, \orgaddress{\city{Moscow}, \country{Russia}}}

\abstract{
We consider the offline problem of parallel relocation of $k$ identical mobile
resources. After each request, known in advance, the resources may move
simultaneously, and the duration of a step is determined by the longest
individual movement. This model is equivalent to the offline version of the
time-optimal $k$-server problem. We prove that the decision version is strongly
NP-complete already on metrics of finite subsets of the Euclidean line, or
equivalently, on vertex metrics of weighted paths. Thus, the computational
hardness persists even under a linear arrangement of the admissible resource
locations. As a positive result, we show that the problem can be solved exactly
in polynomial time on metrics of undirected unweighted graphs with a universal
vertex.
}


\maketitle

\section{Introduction}
\label{sec:introduction}

Consider the problem of planning the relocation of several identical mobile
resources between prescribed positions. Assume that the requests are known in
advance and served in a fixed order, and that the next stage begins only after
all movements of the current stage have been completed. When all resources have
the same speed, the duration of a transition is determined by the longest
individual movement. Other resources can be relocated in parallel without
additional delay. The subject of this paper is the computational complexity of
minimizing the total duration of such transitions.

One applied context for the coordination of mobile resources is robotic
warehousing. For example, in the Kiva system, robots transport mobile racks to
workstations~\cite{wurman2008warehouse}. We do not model such a warehouse in
full, but use it as a relocation abstraction. Unlike lifelong multi-agent
pickup and delivery~\cite{ma2017mapd}, a request here specifies a single point,
and there are no pickup/delivery operations or collision constraints. The fixed
order of requests and synchronization barriers are treated as modeling
assumptions.

Formally, this model is the offline version of the time-optimal $k$-server
problem. The classical $k$-server problem was introduced
in~\cite{manasse1990competitive}. In the standard formulation, the cost is the
sum of movements, and most work has focused on online
algorithms~\cite{koutsoupias2009server}. For this classical sum-distance model,
Bertsimas et al. proposed a mixed-integer formulation of the offline problem and
robust/adaptive approaches to online solutions under uncertainty about future
requests~\cite{bertsimas2019optimization}.

In the time model, resources move in parallel, and the transition cost is the
maximum distance traveled. This variant was studied by Koutsoupias and
Taylor~\cite{koutsoupias2004cnn}. They attribute its original formulation,
under the name min-time, to Fiat, Rabani and Ravid. A modern study of this
model under the name time-optimal $k$-server is presented by Frei et
al.~\cite{frei2025time}. Unlike in the distance model, simultaneous movements
can reduce the cost here, and therefore restricting attention to algorithms
that move only one resource per step is not valid in general.

Here we establish a computational boundary for offline optimization. Unlike the
classical distance model, where an optimal offline schedule can be computed in
polynomial time~\cite{chrobak1991new}, replacing the sum by a maximum makes the
problem strongly NP-hard. This already holds on finite subsets of the Euclidean
line, when the number of resources is part of the input. The admissible positions
are only the points of the given metric space.

A different source of warehouse-related computational hardness has been studied
for picker routing in multi-block warehouses with an unbounded number of
blocks~\cite{prunet2025picker}. There, the objective is a closed route of a
single picker with a selectable order of visited points. Here, the request order
is fixed, and solutions specify parallel relocations of several resources.
Therefore, the picker-routing result does not establish hardness for the problem
considered here.

Together with the negative result, we give one positive boundary. If the metric
is induced by an unweighted graph with a universal vertex, then the exact offline
optimum can be found in polynomial time. In such a network, any position can
reach the common base vertex in one unit of time. The class includes uniform
metrics, unweighted stars, and wheels.

\section{Model}
\label{sec:model}

Let $\mathcal{M}=(X,d)$ be a finite metric space, and let $k\ge1$ be the number
of identical mobile resources. A configuration $C\colon[k]\to X$ specifies their
positions. Several resources may be located at the same point. The input
consists of an initial configuration $C_0$ and a sequence of requests
\[
  \sigma=(r_1,\ldots,r_N), \qquad r_t\in X.
\]
After request $r_t$, a configuration $C_t$ is chosen that contains at least one
resource at point $r_t$. The transition duration is
\[
  \Delta_t=\max_{i\in[k]} d\bigl(C_{t-1}(i),C_t(i)\bigr),
\]
and the total time of the schedule is $\sum_{t=1}^N\Delta_t$. The resources may
be regarded as numbered, but when moving between two multisets of positions, the
matching between them is chosen optimally.

All configurations take values only in $X$. For a weighted path, this means that
after each request the resources are located at vertices of the path, and not at
additional points inside its edges. Partial traversal of a long edge with
continuation in later stages is not allowed. For the time model, this is an
essential restriction~\cite{frei2025time}.

We denote the decision version by \textsc{Offline-Time-$k$-Server}: together
with the data described above, a threshold $L$ is given, and the question is
whether there exists a schedule of duration at most $L$. We assume that the
metric is given explicitly by a matrix of rational distances or by a connected
undirected graph with positive rational edge lengths, and that $L$ is also
rational. Numbers are written in binary, and $C_0$ contains an explicit list of
$k$ positions. Under this encoding, the problem belongs to NP: a certificate
contains $Nk$ indices of resource positions, and feasibility and the sum of
durations can be checked in polynomial time. For $N=0$, the optimum is zero.

\section{Strong NP-Completeness on the Line}
\label{sec:hardness}

We prove the main result by a reduction from the NP-complete problem
\textsc{Vertex Cover}~\cite{karp1972reducibility}. Let its input be a simple
graph $G=(V,E)$ with $|E|=m\ge1$ and an integer $K$, where $0\le K\le m$. This
restriction does not change NP-hardness: instances with $m=0$, $K<0$, or
$K\ge m$ are solved directly. Put $D:=m+1$ and fix orders of the vertices and
edges. For an edge $e=\{u,v\}$, denote its endpoints by $a_e,b_e$ so that
$a_e$ precedes $b_e$ in the vertex order.

\subsection{Construction}

For each edge $e$, create a block of five points
\[
 H_{a_e,e},\quad R_{a_e,e},\quad B_e,\quad R_{b_e,e},\quad H_{b_e,e},
\]
placed in the indicated order on the line. The lengths of the four consecutive
segments are $1,D,D,1$. The points $R_{u,e}$ and $H_{u,e}$ form the incidence
pair $P_{u,e}$. The last reserve point of block $e_i$ is connected to the first
reserve point of block $e_{i+1}$ by a segment of length $2D$. The result is one
weighted path, whose shortest-path metric is also the metric of a finite subset
of the Euclidean line.

Set $k:=4m$ and place one resource at each point of the form $R_{u,e}$ and
$H_{u,e}$. The request sequence starts with all blocker points
\[
  B_{e_1},B_{e_2},\ldots,B_{e_m}.
\]
Then, for each vertex $u$ in the fixed order, its test follows: first all
$H_{u,e}$ for $e\in\delta(u)$, and then all $R_{u,e}$ for $e\in\delta(u)$. The
threshold is
\[
  L:=D+K<2D.
\]
All requests in this construction are pairwise distinct.

We will need only the following immediate properties of the distances. For each
pair $P_{u,e}$, the distance inside the pair is $1$, leaving it for any outside
point requires at least $D$, and leaving it for a point other than the blocker
$B_e$ requires at least $2D$. Any point of another block is at distance at least
$3D$ from $B_e$.

\subsection{Correctness of the Reduction}

\begin{lemma}
\label{lem:forward}
If $A\subseteq V$ is a vertex cover, then the constructed instance can be
served in time at most $D+|A|$.
\end{lemma}

\begin{proof}
For each $e=\{u,v\}$, choose an endpoint $a(e)\in A\cap\{u,v\}$. At the first
request, simultaneously move the resource from $R_{a(e),e}$ to $B_e$ for all
$e$. All these movements have length $D$, so the first transition takes time
$D$. After it, all requests to blockers are already covered.

Until the test of vertex $u$, leave its reserve points unchanged. If $u$ was
not chosen for any incident edge, the whole test is served without movements.
Otherwise, at the first request in the reserve part of the test, simultaneously
move the necessary resources from $H_{u,e}$ to $R_{u,e}$. This takes one unit of
time. The same $u$ is paid for at most once, so the total duration does not
exceed $D+|A|$.
\end{proof}

Call a schedule \emph{short} if its duration is less than $2D$. Any schedule
satisfying the threshold $L$ is short.

\begin{lemma}
\label{lem:invariant}
In a short feasible schedule, after the first request every blocker is occupied
and remains occupied until the end. For each edge $e=\{u,v\}$, after the first
request the number of resources in each of the pairs $P_{u,e}$ and $P_{v,e}$ is
constant. Both pairs are nonempty, and at least one of them contains exactly one
resource.
\end{lemma}

\begin{proof}
The first request to $B_{e_1}$ requires a transition of duration at least $D$.
After it, less than $D$ total time remains, so no later transition can contain a
movement of length at least $D$. If some blocker is not occupied immediately
after the first transition, a later request to it cannot be served. Leaving an
occupied blocker later is also impossible.

The first transition has duration less than $2D$. Therefore a resource from an
incidence pair can either remain inside it or move to its own blocker, but not
to another pair or another block. In each pair, after the first transition at
most two resources remain, and then their number is constant. A pair cannot be
empty: later requests to its points could not be served without entering from
outside. At least one of the four resources of each block must occupy its
blocker. Hence, in the two nonempty pairs together at most three resources
remain, and at least one of them is a singleton.
\end{proof}

\begin{lemma}
\label{lem:reverse}
From any short feasible schedule of duration $F$, one can extract a vertex
cover of size at most $F-D$.
\end{lemma}

\begin{proof}
Include a vertex $u$ in the set $A$ if and only if after the first request some
pair $P_{u,e}$ contains exactly one resource. By Lemma~\ref{lem:invariant}, for
each edge $e=\{u,v\}$ at least one of the pairs $P_{u,e}$ and $P_{v,e}$ is a
singleton. Therefore $A$ is a vertex cover.

If $u\in A$, choose a pair $P_{u,e}$ with a single resource and denote the
indices of its requests by $h_u<r_u$. After the first transition, the resource
cannot leave the pair, and no other resource can enter it. Therefore, at times
$h_u$ and $r_u$, the same resource is located at $H_{u,e}$ and $R_{u,e}$,
respectively. By the triangle inequality,
\[
  \sum_{t=h_u+1}^{r_u}\Delta_t\ge d(H_{u,e},R_{u,e})=1.
\]
These ranges of transitions lie inside the tests of the corresponding vertices,
are pairwise disjoint, and do not contain the first transition. Summing the
inequalities and using $\Delta_1\ge D$, we obtain $F\ge D+|A|$.
\end{proof}

\begin{theorem}
\label{thm:hardness}
\textsc{Offline-Time-$k$-Server} is strongly NP-complete on metrics of finite
subsets of the Euclidean line, or equivalently, on vertex metrics of weighted
paths, when the number of resources is part of the input. The statement remains
true for integer coordinates and pairwise distinct requests.
\end{theorem}

\begin{proof}
Lemma~\ref{lem:forward} gives a schedule of cost at most $D+K$ if the original
graph has a vertex cover of size at most $K$. Conversely, by
Lemma~\ref{lem:reverse}, a schedule of cost at most $D+K<2D$ gives such a
cover. The construction contains $5m$ points, $4m$ resources, and $5m$ requests.
The edge lengths belong to $\{1,D,2D\}$. If the first coordinate is zero, then
all coordinates and distances are bounded by the total length of the path,
\[
  m(2D+2)+(m-1)2D=4m^2+4m-2.
\]
The threshold $L\le2m+1$ is also polynomially bounded. Hence, the hardness is
strong. Membership in NP was noted in Section~\ref{sec:model}.
\end{proof}

\section{A Polynomial Case: Networks with a Common Hub}
\label{sec:positive}

Now consider an unweighted undirected graph $G=(X,E)$ with a universal vertex
$o$ adjacent to all other vertices. In its shortest-path metric, all distances
belong to $\{0,1,2\}$. Such a network models a system in which compatible
states are connected by a direct transition, while any incompatible movement is
available through the common hub $o$.

For a set $A\subseteq X$ of size at most $k$, define the canonical
configuration $K(A)$ as follows. Place one resource at each point of
$A\setminus\{o\}$, and place all remaining resources at $o$. This configuration
covers $A$, including the case $o\in A$. For configurations $C,C'$, let
$\trans(C,C')$ denote the minimum possible duration of their parallel
transition, that is, the bottleneck matching between the two multisets of
positions.

\begin{lemma}
\label{lem:canonical}
Let configurations $C$ and $C'$ cover sets $A$ and $A'$, respectively, both of
size at most $k$, and suppose that $\trans(C,C')>0$. Then
\[
  \trans\bigl(K(A),K(A')\bigr)\le \trans(C,C').
\]
The cost also does not increase if only one of the configurations is replaced
by the corresponding canonical configuration.
\end{lemma}

\begin{proof}
Set $b:=\trans(C,C')\ge1$ and fix a matching of movements of length at most
$b$. In each configuration being canonicalized, keep one copy of each required
point different from $o$, and replace the remaining positions by $o$. If both
endpoints of a matched pair remain unchanged, then its length is still at most
$b$. Otherwise, at least one endpoint has become $o$, and therefore the new
length is at most $1\le b$. The obtained matching is feasible after replacing
one or both configurations, which proves all claims.
\end{proof}

\begin{theorem}
\label{thm:poly}
On metrics of unweighted graphs with a universal vertex, the problem
\textsc{Offline-Time-$k$-Server} can be solved in polynomial time.
\end{theorem}

\begin{proof}
Regard configurations as multisets, and separate a possibly empty initial
prefix served without changing $C_0$. Split the remaining sequence in an
optimal schedule into maximal intervals with a constant configuration. Each
such interval $I$ contains at most $k$ distinct requested points, forming a set
$A(I)$. All original transitions between neighboring intervals, as well as the
first transition from $C_0$, are positive.

Simultaneously replace the configurations of these intervals by $K(A(I))$. For
each transition between intervals, apply Lemma~\ref{lem:canonical} to the pair
of original configurations. For the first transition from $C_0$, use its
one-sided statement. Transitions inside intervals remain zero. Thus, the cost
does not increase, even if some new transitions between intervals have become
zero.

Such a partition is found by shortest-path dynamic programming on an acyclic
interval graph. Its vertices are all intervals $I=[p,q]$ for which
$|A(I)|\le k$. An edge from $[p,q]$ to $[q+1,r]$ has weight
\[
  \trans\bigl(K(A([p,q])),K(A([q+1,r]))\bigr).
\]
The source is connected to $[p,q]$ if the initial configuration covers requests
$1,\ldots,p-1$. The weight of such an edge is
$\trans(C_0,K(A([p,q])))$. A shortest path to an interval ending at $N$ gives
the exact optimum: every obtained canonical schedule defines such a path, and
every path is realized by movements at the beginnings of its intervals and zero
transitions inside them. The case in which all requests are already covered by
$C_0$, including the empty sequence, is checked separately.

There are $O(N^2)$ intervals and $O(N^3)$ transitions. To compute $\trans$,
first check equality of the multisets, which gives cost $0$. Otherwise, the
cost is $1$ if there is a perfect matching in the bipartite graph of pairs of
positions at distance at most $1$, and is $2$ otherwise. Coincident positions
are represented by different vertex copies. Finding a matching by successive
augmenting paths takes $O(k^3)$ time. Taking the preparation of the adjacency
matrix into account, the final bound is $O(|X|^2+N^3k^3)$ for $N\ge1$, which is
polynomial in the input size.
\end{proof}

\section{Conclusion}

Offline planning of parallel resource relocation becomes strongly NP-hard
already on finite subsets of the line. The simple geometry of admissible
positions does not remove the difficulty of minimizing the sum of durations of
parallel stages. At the same time, for networks with a common hub, there is an
exact polynomial-time algorithm. Natural next questions include approximation
algorithms for the linear case and extending the polynomial boundary to broader
classes of graph metrics.

\bibliography{offline_references}
\end{document}